\documentclass{llncs}
\usepackage{graphicx,amssymb,amsmath}
\usepackage{enumitem}

\newcommand{\opt}{\ensuremath{\mathop{OPT}}}

\usepackage{algorithm}
\usepackage[noend]{algpseudocode}

\title{Minimizing the Total Movement for Movement to Independence Problem on a Line}

\author{Mehrdad Ghadiri\thanks{Computer Engineering department,
       		Sharif University of Technology, {\tt ghadiri@ce.sharif.edu}}
        \and
        Sina Yazdanbod \thanks{Computer Engineering department,
       		Sharif University of Technology, {\tt  syazdanbod@ce.sharif.edu}}}

\institute{Sharif University of Technology}
\index{Yazdanbod, Sina}
\index{Ghadiri, Mehrdad}

\begin{document}
\thispagestyle{empty}
\maketitle

\begin{abstract}
Given a positive real value $\delta$, a set $P$ of points along a line and a distance function $d$, in the movement to independence problem, we wish to move the points to new positions on the line such that for every two points $p_{i},p_{j} \in P$, we have $d(p_{i},p_{j}) \geq \delta$ while minimizing the sum of movements of all points. This measure of the cost for moving the points was previously unsolved in this setting. However for different cost measures there are algorithms of $O(n \log(n))$ or of $O(n)$. We present an $O(n \log(n))$ algorithm for the points on a line and thus conclude the setting in one dimension.

\end{abstract}

\section{Introduction}

The problem of minimizing the movement of points to reach a property was introduced first by Demaine et~al.~\cite{DeMv}, which was for the most part in graphical settings.
Many applications appear for the minimizing movement problem is in the contexts of reliable radio networks~\cite{bredin2005deploying}, \cite{corke2004autonomous}, robotics~\cite{Hsiang} and map labeling~\cite{Jiang}\cite{Doddi}.
 In simple terms, the problem of movement to independence on graphs is defined as given a graph $G$ and a set of pebbles $P$, move the pebbles such that no two pebbles occupy the same vertex. They considered the Total Sum measure on different problems. Although, they proved different $\mathrm{NP}$-completeness results for other problems, the problem of whether the movement to independence problem with Total Sum measure is $\mathrm{NP}$-complete, remains open to this day.
Time complexity of the algorithms given in \cite{DeMv} were polynomial in the number of vertices. However, the number of pebbles can be much smaller than the number of vertices of the graph. That is why in \cite{DeTr}, they turned to fixed-parameter tractability. Dumitrescu et~al.~\cite{DuCk} were the first to consider the settings of a real line. They gave LP-based algorithms for movement to independence on a line and on a closed curve with the measure of minimizing the maximum movement of points. In closed-curve version of the problem, authors defined distance as the length of the smallest subcurve between two points. Dumitrescu et~al.~\cite{DuCk}'s algorithms for both real line settings and closed curve settings were recently improved by Li et~al.~\cite{LiMm} with a linear time algorithm.
Our contribution in this paper is considering the problem of Total Sum on the same settings of \cite{DuCk}. 

The rest of this paper is structured as follows. In Section 2, we explain preliminaries and the definitions of our problem. In Section 3, the formal settings of the problem is presented. The algorithm and its proof are written in Section 4 and the $O(nlogn)$ implementation and complexity analysis of it are presented in Section 5. In the end, we conclude the article in the last section and give an open problem for further research.
 
\section{Preliminaries}
In movement to independence problem, we are given a positive real value $\delta$, a set $P$ of points and a distance function $d$ and we wish to move the points to new positions such that any two points are at least $\delta$ apart. The goal is to minimize this movement. There are several different measures of movement. We consider the \emph{$\mathit{\mathop{TotalSum}}$ measure} which is the sum of movements of all points. We examine this problem in the setting of real line. In this section, we define our the terminology and introduce the problems considered in this paper.

\begin{definition}[Configuration]
For a set of points $P$, we define a \emph{configuration} $H$ of $P$ to be a placement of points in the domain. For a point $p \in P$, we use $H(p)$ to denote the location of $p$ in configuration $H$.
\end{definition}

In this paper, we will investigate movement to independence problem in the setting defined bellow.

\begin{definition}[Independence]
Given a set of points $P$, a positive value $\delta \in \mathbb{R}^+$ and a distance function $d$, a configuration $H$ is called independent, whenever for every two points $p_{i},p_{j}\in P$, we have $d(H(p_{i}), H(p_{j})) \geq \delta$.
\end{definition}

The formal definition of the movement to independence problem with total sum measure is as follows.
\begin{definition}
Let $P = \{p_1, \ldots, p_n\}$ be a set of points and $I$ be its initial configuration.
Given $\delta \in \mathbb{R}^+$ and a distance function $d$,
find an independent configuration $F$ of $P$, so as to minimizes
$\sum_{i=1}^n d(I(p_i), F(p_i))$.
\end{definition}

The point set $P$ can be from different domains. In the following, we define the distance function used for these domains.

\begin{definition}[On a Line]
For two points $p_i, p_j \in P$ and configurations $H$ and $H'$ (not necessarily different) of points $P$ on the real line, we define distance as $d(H(p_i), H'(p_j)) = |H(p_i) - H'(p_j)|$.
\end{definition}

In our algorithm, we make use of chains of points. In a linear domain, a set of points in a configuration form a chain, whenever they are tightly put together in 
distances of $\delta$.

\begin{definition}
In a configuration $H$ of points $P$,
we call a subset $C = \{q_1, \ldots, q_j\}$ of $P$, where $H(q_1) < \cdots < H(q_j)$, a \emph{chain} in $H$,
 if we have $d(H(q_{i}), H(q_{i+1}))= \delta$ for all $i = 1, \ldots, j-1$ (see Figure~\ref{chain}).
\end{definition}

A chain is maximal if it is not a proper subset of another chain.
Unless noted explicitly, we consider chains to be maximal.
Chain partitioning is the act of partitioning independent configuration into maximal chains. Figure~\ref{chain} shows an example of this partitioning. In this figure the rectangles show the chains.

\begin{figure}[ht]
	\centering
	\includegraphics[width=0.8\columnwidth]{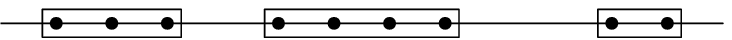}
	\caption{Partitioning $H$ into chains}
	\label{chain}
\end{figure}

\section{Setting of a Real Line}
In this section we study the problem of movement to independence in real line domain with total sum measure.
Let a point set $P = \{p_1, \ldots, p_n\}$ be in $\mathbb{R}$ with initial configuration $I$. For the sake of simplicity, we assume that the input is given in sorted order and that points are initially in distinct positions. That is $I(p_i) < I(p_j)$ for every $i < j$. The following lemma shows that in fact we can make such assumptions.

\begin{lemma}
\label{lemma:lineorder}
The initial order of points $P$ is preserved in the optimal configuration $\opt$. In other words, an optimal configuration $\opt$ exists in which for any two points $p_i, p_j \in P$ with $I(p_i) < I(p_j)$, we have $\opt(p_i) < \opt(p_j)$.
\end{lemma}

\begin{proof}
Let $\opt$ be an optimal configuration. Assume that there exist $p_i, p_j \in P$ violating this. Then, we can swap the location of $p_i$ and $p_j$ in $\opt$, as in Figure~\ref{order}, without increasing the total movement of points. Based on symmetry, one should only consider the three situations shown in Figure~\ref{order}.
\end{proof}
\begin{figure}[ht]
	\centering
	\includegraphics[width=\columnwidth]{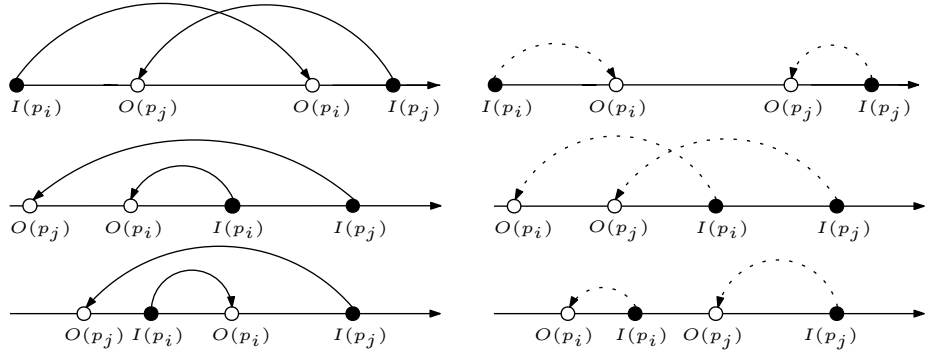}
	\caption{Three different situations of unordered points that can be put in order, without increasing total sum. Figures on the left shows the problem and figures on the right show the reordered version without moving the points}
	\label{order}
\end{figure}

In a given configuration $H$, We define the sets $L_H(S)$, $R_H(S)$ and $O_H(S)$ as follows.
\begin{definition}
For a given configuration $H$ of points $P$ with initial configuration $I$. For a subset $S \subseteq P$.
\begin{itemize}
\item $L_H(S) = \{ p \in S \mid H(p) < I(p) \}$
\item $R_H(S) = \{ p \in S \mid H(p) > I(p) \}$
\item $O_H(S) = \{ p \in S \mid H(p) = I(p) \}$
\end{itemize}
We may use the notations $L(S)$, $R(S)$ and $O(S)$ instead, when there is no ambiguity.
\end{definition}

Our algorithm for real line domain is iterative and adds one point at a time until all the points are inserted. Let $I$ be the initial configuration of points $P$ and $H_i$ be the configuration generated by our algorithm at the end of $i$-th iteration, which is an independent configuration (Note that $H_i$ is defined over the set $\{p_1, \ldots, p_i\}$). Let $C =\{c_1, \ldots, c_k\}$ be the sorted set of chain partitioning of $H_i$, where $c_k$ is the rightmost chain.

The main idea in this algorithm is that at the end of each iteration of the algorithm, the following properties are preserved for every chain $c_{j} \in C$:
\begin{property}
\label{prop:p1}
For every chain $c_j \in C$, we have
$|L_{H_i}(c_j)| + |O_{H_i}(c_j)| > |R_{H_i}(c_j)|$
\end{property}
\begin{property}
\label{prop:p2}
For every chain $c_j \in C$, we have
$|L_{H_i}(c_j)| \leq |O_{H_i}(c_j)| + |R_{H_i}(c_j)|$
\end{property}

Any configuration with these two properties is ,in a sense, locally optimal. We show that our algorithm creates the optimal solution with these properties.

\section{The Algorithm}

Our algorithm starts from $p_1$ and at each iteration adds one new point to the configuration.
Each iteration of the algorithm consists of two phases.
In the first phase, we insert a new point into the current configuration and if that violates Property~\ref{prop:p1}, in the second phase, we restore that property. In the following, we explain the procedure for each phase.

\subsection{Phase 1}
Assume that we have the configuration $H_i$ after processing the first $i$ points and let $p_{i+1}$ be the first remaining point and $\{c_1,\ldots,c_k\}$ be the chain partitioning of $H_i$. In this phase, we construct configuration $G_{i+1}$ over the set $\{p_1, \ldots, p_{i+1}\}$, let $G_{i+1}(p_j) = H_i(p_j)$, for all $j=1, \ldots, i$. So, we just need to determine the location of $p_{i+1}$ in $G_{i+1}$.
We consider the two following cases based on the distance of the new point from the last inserted point in our created configuration:

\begin{enumerate}[leftmargin=\parindent,align=left,itemindent=\parindent,label={\textbf{Case \arabic*}. }]
\item If $d(H_i(p_i),I(p_{i+1})) \geq \delta$ and $I(p_{i+1})$ is to the right of $H_i(p_i)$, then we set $G_{i+1}(p_{i+1}) = I(p_{i+1})$. If $d(G_{i+1}(p_i), G_{i+1}(p_{i+1})) > \delta$, then the chain partitioning of $G_{i+1}$ is $\{c_1,\ldots,c_k,c_{k+1}\}$ where $c_{k+1}$ is a new chain which its only member is $p_{i+1}$ (Figure~\ref{cone}) . If $d(G_{i+1}(p_i), G_{i+1}(p_{i+1})) = \delta$, the chain partitioning of $G_{i+1}$ is $\{c_1,\ldots,c_{k-1},c'_k\}$ where $c'_k$ is a new chain which is $c_{k} \cup p_{i+1}$ (Figure~\ref{ctwo}). Clearly, the new point is in $O_{c'_{k}}$( the set of points from $c'_{k}$ that are on their initial location) or $O_{c_{k+1}}$. That is, the resulting configuration preserves Property~\ref{prop:p1} and \ref{prop:p2}. Therefore, there is no need to run Phase~2 and we proceed to the next iteration. Note that in this case $H_{i+1}$ will be $G_{i+1}$.

\begin{figure}[ht]
	\centering
	\includegraphics[width=0.8\columnwidth]{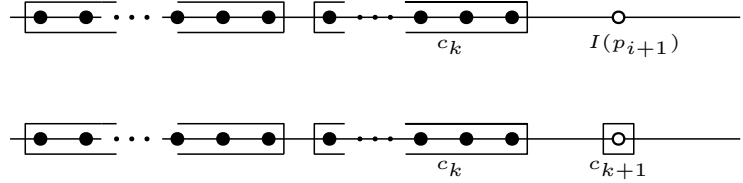}
	\caption{This is the case where $d(H_i(p_i),I(p_{i+1})) > \delta$. The new point will create a chain consisting of itself}
	\label{cone}
\end{figure}

\begin{figure}[ht]
	\centering
	\includegraphics[width=0.8\columnwidth]{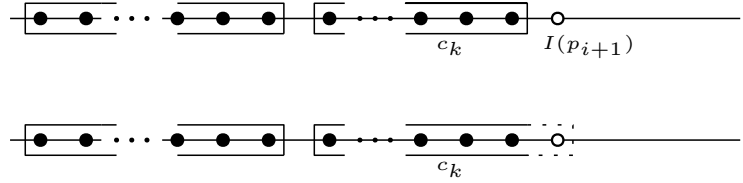}
	\caption{This is the case where $d(H_i(p_i),I(p_{i+1})) = \delta$. The new point will merge with the previous chain}
	\label{ctwo}
\end{figure}

\item If $d(H_i(p_i),I(p_{i+1}))<\delta$ (Figure~\ref{cthree}) or $d(H_i(p_i),I(p_{i+1})) \geq \delta$ and $I(p_{i+1})$ is to the left of $H_i(p_i)$ (Figure~\ref{cfour}), we set $G_{i+1}(p_{i+1})$ to $H_i(p_i)+\delta$. Therefore, the chain partitioning of $G_{i+1}$ is $\{c_1,\ldots,c_{k-1},c'_k\}$ where $c'_k$ is a new chain which is $c_{k} \cup p_{i+1}$. The only complication is that Property~\ref{prop:p1} might get violated in $G_{i+1}$ because $p_{i+1} \in R_{G_{i+1}}(c'_k)$.
In this case, we proceed to Phase~2, to move the chains so that Property~\ref{prop:p1} is restored again. Otherwise, there is no need to run Phase~2 and we proceed to the next iteration. Note that in this case $H_{i+1}$ will be $G_{i+1}$.
\end{enumerate}

\begin{figure}[ht]
	\centering
	\includegraphics[width=0.8\columnwidth]{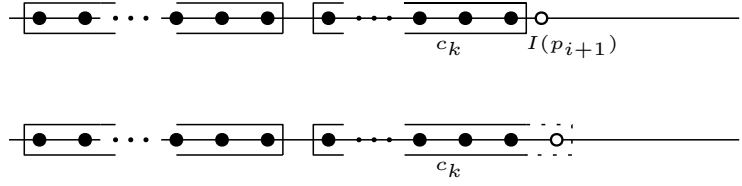}
	\caption{This is the case where $d(H_i(p_i),I(p_{i+1})) < \delta$ but the new point is still located after $H_i(p_i)$. }
	\label{cthree}
\end{figure}

\begin{figure}[ht]
	\centering
	\includegraphics[width=0.8\columnwidth]{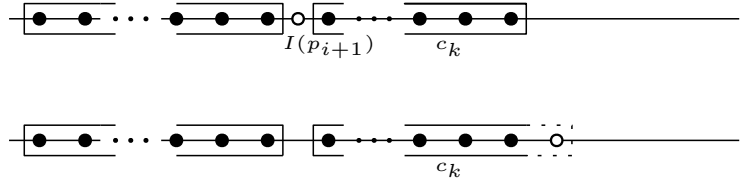}
	\caption{This is the case where $d(H_i(p_i),I(p_{i+1})) \geq \delta$ but the new point is still located before $H_i(p_i)$.}
	\label{cfour}
\end{figure}

\subsection{Phase 2}
In this phase, we construct the configuration $H_{i+1}$ given the configuration $G_{i+1}$ from the previous phase which its chain partitioning is $\{c_1,\ldots,c_{k-1},c'_k\}$. For configuration $H_{i+1}$, we have $H_{i+1}(p_j) = G_{i+1}(p_j)$, if $p_j \in c_1 \cup \cdots \cup c_{k-1}$ and we just need to determine location of points in $c'_k$ in configuration $H_{i+1}$. The reason for running this phase is that Property~\ref{prop:p1} is violated for the chain $c'_k$ in Phase~1, which means
$|R_{G_{i+1}}(c'_k)| \geq |L_{G_{i+1}}(c'_k)| + |O_{G_{i+1}}(c'_k)|$.
Since
$|R_{G_{i+1}}(c'_k)| = |R_{H_i}(c_k)| + 1$
and
$|R_{H_i}(c_k)| < |L_{H_i}(c_k)| + |O_{H_i}(c_k)|$,
we can infer that
$|R_{G_{i+1}}(c'_k)| = |L_{G_{i+1}}(c'_k)| + |O_{G_{i+1}}(c'_k)|$.
It is clear that Property~\ref{prop:p1} and \ref{prop:p2} still hold for the other chains and also Property~\ref{prop:p2} holds for $c'_k$.
To restore Property~\ref{prop:p1} for $c'_k$, we do as follows.

Let
$\alpha = \min_{p_j \in R_{c'_{k}}} | d(G_{i+1}(p_j),I(p_j))|$ be the value of the minimum distance between a point in the new configuration and their initial configuration. 
and
$\beta = d(G_{i+1}(p_r),G_{i+1}(p_l)) - \delta$,
where $p_r$ is the rightmost point of $c_{k-1}$ in configuration $G_{i+1}$ and $p_l$ is the leftmost point of $c'_k$ in that configuration. In other words, $\beta + \delta$ is the distance of the last point of the chain $c_{k-1}$ and first point of the chain $c'_{k}$. We consider the two following cases:

\begin{enumerate}[leftmargin=\parindent,align=left,itemindent=\parindent,label={\textbf{Case \arabic*}. }]
\item If $\alpha < \beta$, we set $H_{i+1}(p_j) = G_{i+1}(p_j) - \alpha$, for all $p_j \in c'_k$. It is clear that we have $|R_{G_{i+1}}(c'_k)| > |R_{H_{i+1}}(c'_k)|$. Therefore $|L_{H_{i+1}}(c'_k)| + |O_{H_{i+1}}(c'_k)| > |R_{H_{i+1}}(c'_k)|$. We also know that $L_{G_{i+1}}(c'_k) \cup O_{G_{i+1}}(c'_k) = L_{H_{i+1}}(c'_k)$ and $R_{G_{i+1}}(c'_k) = O_{H_{i+1}}(c'_k) \cup R_{H_{i+1}}(c'_k)$. Since $|R_{G_{i+1}}(c'_k)| = |L_{G_{i+1}}(c'_k)| + |O_{G_{i+1}}(c'_k)|$, we conclude that $|L_{H_{i+1}}(c'_k)| \leq |O_{H_{i+1}}(c'_k)| + |R_{H_{i+1}}(c'_k)|$.
Therefore, Property~\ref{prop:p1} is restored and Property~\ref{prop:p2} is preserved.

\item If $\alpha \geq \beta$, we set $H_{i+1}(p_j) = G_{i+1}(p_j) - \beta$, for all $p_j \in c'_k$. It is clear that $c_{k-1}$ and $c'_k$ are not maximal chains in $H_{i+1}$. Therefore, the chain partitioning of $H_{i+1}$ is $\{c_1\ldots,c_{k-2},c'_{k-1}\}$, where $c'_{k-1}=c_{k-1} \cup c'_{k}$. We can say two chains $c_{k-1}$ and $c'_k$ are merged (Figure \ref{merge}). We have $|L_{H_{i+1}}(c'_k)| + |O_{H_{i+1}}(c'_k)| \geq |R_{H_{i+1}}(c'_k)|$, because $|L_{G_{i+1}}(c'_k)| + |O_{G_{i+1}}(c'_k)| = |R_{G_{i+1}}(c'_k)|$. We also have $|L_{H_{i+1}}(c_{k-1})| + |O_{H_{i+1}}(c_{k-1})| > |R_{H_{i+1}}(c_{k-1})|$. Hence, we will have $|L_{H_{i+1}}(c'_{k-1})| + |O_{H_{i+1}}(c'_{k-1})| > |R_{H_{i+1}}(c'_{k-1})|$ and Property~\ref{prop:p1} is restored. By a reasoning similar to previous case, we can conclude that Property~\ref{prop:p2} holds for $c'_{k-1}$ in $H_{i+1}$.
\end{enumerate}

\begin{figure}[ht]
	\centering
	\includegraphics[width=0.8\columnwidth]{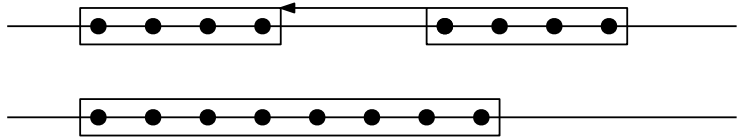}
	\caption{When the left chain reaches $\delta$ radius of the right chain, they merge.}
	\label{merge}
\end{figure}

\subsection{Correctness}

We claim that this algorithm returns an optimal configuration. But before we go on to prove that claim, we state a lemma.

\begin{lemma}
\label{lem:maximal-chains}
Let $c = \{p_i, \ldots, p_j\}$ be a maximal chain in $H_n$. For all $l$, where $i \leq l \leq j$, for the non-maximal chain $c' = \{p_i, \ldots, p_l\}$, we have
\[
|L_{H_n}(c')| + |O_{H_n}(c')| > |R_{H_n}(c')|.
\]
\end{lemma}

\begin{proof}
In the $l$-th iteration of the algorithm, $p_l$ was inserted into the configuration. Let $\{c_1,\ldots,c_k\}$ be the chain partitioning of $H_l$ and $c'=c_m \cup \cdots \cup c_k$. We have $|L_{H_l}(c')| + |O_{H_l}(c')| > |R_{H_l}(c')|$, because Property~\ref{prop:p1} holds for all chains in $\{c_1,\ldots,c_k\}$. After iteration $l$, points in the $c'$ only move leftwards or does not move in each iteration. Thus, the left side of the inequality is non-decreasing and the right side is non-increasing. Therefore, at the end of every further iteration, the inequality still holds. In particular, the inequality holds at the end of the algorithm.
\end{proof}

Now we have the sufficient tools to prove optimality of the output of this algorithm. We make the argument in two cases, once we take the rightmost difference from the optimal and second we use the leftmost difference. In the end, our solution is optimal or simply a shift of the optimal to the right that does not increase the sum. 
\begin{theorem}
\label{1}
Configuration $H_n$ is optimal.
\end{theorem}

\begin{proof}
Let $\opt$ be an optimal configuration of points $P$ which preserves the order of initial configuration. According to Lemma~\ref{lemma:lineorder}, this configuration exists.
Take the rightmost point $p_l$ in $H_n$ such that $\opt(p_l) < H_n(p_l)$ (Figure~\ref{left}). Assume that this point is in the chain $c=\{p_i, \ldots, p_j\}$. Figure~\ref{left} depicts this situation.

\begin{figure}[ht]
	\centering
	\includegraphics[width=0.8\columnwidth]{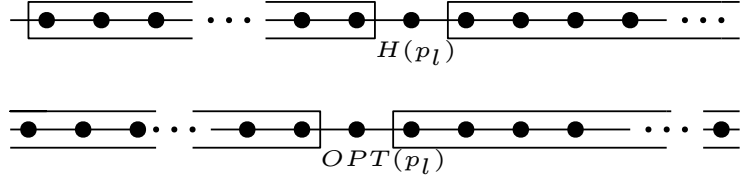}
	\caption{The above figure is the chain containing $p_{l}$ in our solution while below figure show the same points of the above chain in the optimal solution. This chain is not necessarily intact, and may have been divided into several different chains in the optimal solution.  }
	\label{left}
\end{figure}

We know from Lemma~\ref{lem:maximal-chains} that for the $c' = \{p_i, \ldots, p_l\}$ in $H_n$, we have:
\[
|L_{H_n}(c')| + |O_{H_n}(c')| > |R_{H_n}(c')|
\]
For each $k = i, \ldots, l$ we have $H_n(p_k) < \opt(p_k)$. Since, the order of points in $\opt$ is like $H_n$ and also $\opt$ is an independent configuration. Therefore, $O_{H_n}(c') \cup L_{H_n}(c') \subset L_{\opt}(c')$ and $ |L_{\opt}(c')| > |O_{\opt}(c')| + |R_{\opt}(c')|$.
Hence, if we shift the points of $c'$ in $\opt$ to right by $d(H_n(p_l), \opt(p_l))$, the number of points getting further away from their initial location will be smaller than the number of points getting closer and also the points will remain independent from each other. Hence, total movement of points will decrease, which contradicts the optimality of $\opt$. Therefore, there are no points in the optimal configuration to the left of their corresponding point in $H_n$.

On the other hand, let $H_n(p_l)$ be the leftmost point such that $\opt(p_l) > H_n(p_l) $(Figure~\ref{left}). Assume that this point is in the chain $c=\{p_i, \ldots, p_j\}$ in $H_n$. This case is shown in Figure~\ref{right}.

\begin{figure}[ht]
	\centering
	\includegraphics[width=0.8\columnwidth]{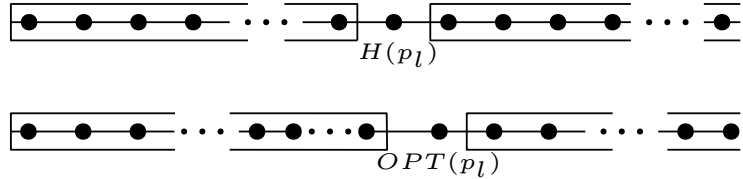}
	\caption{Same as before, rectangles are chains and $p_{l}$ is the first difference.}
	\label{right}
\end{figure}

This time, we use Property~\ref{prop:p2}. For the chain $c$ we have:
\[
|L_{H_n}(c)| \leq |O_{H_n}(c)| + |R_{H_n}(c)|.
\]
It is fairly easy to see that:
\[
|L_{\opt}(c)| \leq |L_{H_n}(c)| \leq |O_{H_n}(c)| + |R_{H_n}(c)|
\]
giving
\[
|O_{H_n}(c)| + |R_{H_n}(c)| \leq |O_{\opt}(c)| + |R_{\opt}(c)| 
\]
because the order of points in $\opt$ is like $H_n$ and also $\opt$ is an independent configuration. If we shift the points $p_l, \ldots, p_j$ in the configuration $\opt$ to left, total movement of points will not increase until after $p_l$ coincides with $H_n(p_l)$. That is to say, the total movement of the solution returned by our algorithm is less than or equal to that of the optimal configuration. After placing $p_l$ on $H_n(p_l)$ by moving all the points $\{p_l, \ldots, p_j\}$ in the configuration $\opt$ to left, we have a new configuration $\opt'$ with the same (if not less) total movement. Now, we find the next point from $\opt'$ with this property (leftmost point such that $\opt'(p_l) > H_n(p_l)$) and we continue until all the points with this property are converted to their corresponding point in $H_n$, therefore, proving that the cost of our solution is at most that of the optimal solution.
\end{proof}

A naive implementation of this algorithm runs in $O(n^2)$ time. However, in Section~\ref{sec:imp} we give a more efficient implementation that runs in $O(n \log(n))$ time.

\section{Implementaion and complexity analysis}
\label{sec:imp}
In each iteration of the algorithm, there are two phases. In the first one, a point is placed on its initial location or on the end of the last chain. Obviously, the complexity of this phase is $O(1)$. In the second phase, we move all of a chain and possibly merge it with another chain. If we update location of all the points in this phase then in each iteration, the complexity of this phase is as the size of the moving chain. Therefore, in the worst case, the complexity of the algorithm will be $O(n^2)$. We use a little trick to reduce the complexity of the algorithm. Let $c_1, \ldots, c_k$ be the chains in a configuration like $H$. Let $r(c_j)$ be a real number that shows the total movement of $c_j$ to the left since it was created. In other words, $r(c_j)$ is the total movement of the left most point of $c_j$ to the left since it has been added to the $c_j$.

Let $p$ be a newly added point to the $c_k$ and its location be $\ell$. The trick is that instead of storing the actual location of $p$, we store $\ell - r_j$. When we need the actual location of $p$, we can easily recompute that. Also, when we move the chain to the left, it is sufficient to update just $r_j$ and we do not need to update a number for each point. With this trick, we reduce the time complexity of moving the chains to $O(1)$. There are two other things that affect the complexity of the algorithm: finding the amount of movement of a chain in the second phase of each iteration and merging two chains when we deal with Case~2 of the second phase.

For finding the amount of movement of a chain, we can store all the right points of a chain in a min-heap according to their distances to their initial locations. When we add a point to the chain, we can easily add it to the heap in $O(\log n)$ and when we move the chain to the left, we need to remove the point that is locates on its initial location from the heap which can be done in $O(\log n)$.

In case of merging the chains, let $c_j$ and $c_{j+1}$ be the chain that merged and the new chain be $c'$. We need to merge $c_j$ and $c_{j+1}$'s heaps which can be done in $O(\log n)$ if we use a binomial heap\cite[p.~462]{DBLP:books/mg/CormenLRS01}. The other thing that we need to do  is to set a value for r(c'). To do this we choose one of $c_j$ or $c_{j+1}$ that have more points and set $r(c')$ as its $r$ value and update the location value of the points of the other chain using $r(c')$. Note that, the new value of $r(c')$ does not show necessarily the amount of movement of the new chain but it can be treated as before. The amortized cost of this action is $O(\log n)$ like the disjoint-set data structure~\cite[p.~504]{DBLP:books/mg/CormenLRS01}.

Due to the above analysis, we can conclude that the cost of each iteration of the algorithm is $O(\log n)$ and the complexity of the algorithm is $O(n\log n)$.

\begin{theorem}
\label{2}
Running time of the algorithm is $O(n\log(n))$
\end{theorem}

\section{Conclusion}
In this paper we considered the problem of minimizing total sum of movement of points to reach independence and presented an $O(nlogn)$ algorithm. While the problem for minimizing the movement of point on a circle( or a closed curve) still remains unsolved. It is easy to see that our properties determining a local optimal can be considered in the circle case as well. However, this problem shows to be a little more trickier to solve and these properties might not be enough.

\section{Acknowledgment}
In the end, we would like to thank our dear friend, Sahand Mozaffari, for his thoughtful comments and suggestions.

\small
\baselineskip= 0.85\baselineskip
\bibliographystyle{abbrv}
\bibliography{LNCS.bib}

\begin{thebibliography}{10}

\bibitem{bredin2005deploying}
J.~L. Bredin, E.~D. Demaine, M.~Hajiaghayi, and D.~Rus.
\newblock Deploying sensor networks with guaranteed capacity and fault
  tolerance.
\newblock In {\em Proceedings of the 6th ACM international symposium on Mobile
  ad hoc networking and computing}, pages 309--319. ACM, 2005.

\bibitem{corke2004autonomous}
P.~Corke, S.~Hrabar, R.~Peterson, D.~Rus, S.~Saripalli, and G.~Sukhatme.
\newblock Autonomous deployment and repair of a sensor network using an
  unmanned aerial vehicle.
\newblock In {\em Robotics and Automation, 2004. Proceedings. ICRA'04. 2004
  IEEE International Conference on}, volume~4, pages 3602--3608. IEEE, 2004.

\bibitem{DBLP:books/mg/CormenLRS01}
T.~H. Cormen, C.~E. Leiserson, R.~L. Rivest, and C.~Stein.
\newblock {\em Introduction to Algorithms, Second Edition}.
\newblock The {MIT} Press and McGraw-Hill Book Company, 2001.

\bibitem{DeMv}
E.~D. Demaine, M.~T. Hajiaghayi, H.~Mahini, A.~S. Sayedi{-}Roshkhar, S.~O.
  Gharan, and M.~Zadimoghaddam.
\newblock Minimizing movement.
\newblock {\em {ACM} Transactions on Algorithms}, 5(3), 2009.

\bibitem{DeTr}
E.~D. Demaine, M.~T. Hajiaghayi, and D.~Marx.
\newblock Minimizing movement: Fixed-parameter tractability.
\newblock {\em {ACM} Transactions on Algorithms}, 11(2):14:1--14:29, 2014.

\bibitem{Doddi}
S.~Doddi, M.~V. Marathe, A.~Mirzaian, B.~M.~E. Moret, and B.~Zhu.
\newblock Map labeling and its generalizations.
\newblock In {\em Proceedings of the Eighth Annual {ACM-SIAM} Symposium on
  Discrete Algorithms, 5-7 January 1997, New Orleans, Louisiana.}, pages
  148--157, 1997.

\bibitem{DuCk}
A.~Dumitrescu and M.~Jiang.
\newblock Constrained k-center and movement to independence.
\newblock {\em Discrete Applied Mathematics}, 159(8):859--865, 2011.

\bibitem{Hsiang}
T.~Hsiang, E.~M. Arkin, M.~A. Bender, S.~P. Fekete, and J.~S.~B. Mitchell.
\newblock Algorithms for rapidly dispersing robot swarms in unknown
  environments.
\newblock In {\em Algorithmic Foundations of Robotics V, Selected Contributions
  of the Fifth International Workshop on the Algorithmic Foundations of
  Robotics, {WAFR} 2002, Nice, France, December 15-17, 2002}, pages 77--94,
  2002.

\bibitem{Jiang}
M.~Jiang, J.~Qian, Z.~Qin, B.~Zhu, and R.~J. Cimikowski.
\newblock A simple factor-3 approximation for labeling points with circles.
\newblock {\em Inf. Process. Lett.}, 87(2):101--105, 2003.

\bibitem{LiMm}
S.~Li and H.~Wang.
\newblock Algorithms for minimizing the movements of spreading points in linear
  domains.
\newblock In {\em Proceedings of the 27th Canadian Conference on Computational
  Geometry, {CCCG} 2015, Kingston, Ontario, Canada, August 10-12, 2015}, 2015.

\end{thebibliography}

\end{document}